\documentclass[prx]{revtex4-1}
\usepackage{epsfig,amsmath,amsfonts,amssymb,amstext,afterpage,psfrag,%slashbox
}
\usepackage{slashed}
\usepackage{multirow}
\usepackage{hyperref}
\usepackage{color}
\usepackage[utf8x]{inputenc}
\usepackage{graphicx}
\usepackage{subcaption}
\usepackage{mathtools}
\usepackage{todonotes}
\usepackage[toc,page]{appendix}
\usepackage[normalem]{ulem}
\usepackage{amsthm}

\usepackage{array}

\newcolumntype{C}[1]{>{\centering\let\newline\\\arraybackslash\hspace{0pt}}p{#1}}

\newcommand{\tf}{{\tt TensorFlow}}
\newcommand{\veg}{{\tt VEGAS}}
\newcommand{\foam}{{\tt Foam}}
\newcommand{\iflow}{{\tt i-flow}}

\definecolor{orange}{RGB}{255,127,0}

\newtheorem{theorem}{Theorem}
\newtheorem{corollary}{Corollary}

\begin{document}

\begin{titlepage}

\begin{flushright}
  {\small
    FERMILAB-PUB-20-010-T\\
\today
}
\end{flushright}

\vspace{0.5cm}
\begin{center}
{\Large\bf \boldmath \iflow : High-dimensional Integration and Sampling with Normalizing Flows
%\vspace*{0.3cm}                                                            
\unboldmath}
\end{center}

\vspace{0.5cm}
\begin{center}
{\sc Christina Gao$^{1}$, Joshua Isaacson$^{1}$, and Claudius Krause$^{1}$} 
\end{center}

\vspace*{0.4cm}

\begin{center}
  $^{1}$ Theoretical Physics Department, Fermi National Accelerator Laboratory, Batavia, IL, 60510, USA\\
  %\vspace*{0.2cm}
\end{center}

\vspace{1.5cm}
\begin{abstract}
\vspace{0.2cm}\noindent
In many fields of science, high-dimensional integration is required. Numerical methods have been developed to evaluate these complex integrals. We introduce the code \iflow, a python package that performs high-dimensional numerical integration utilizing normalizing flows. Normalizing flows are machine-learned, bijective mappings between two distributions. \iflow\ can also be used to sample random points according to complicated distributions in high dimensions. We compare \iflow\ to other algorithms for high-dimensional numerical integration and show that \iflow\ outperforms them for high dimensional correlated integrals. The \iflow\ code is publicly available on gitlab at {\tt https://gitlab.com/i-flow/i-flow}. 

\end{abstract}
\maketitle
\vfill

\end{titlepage}

%= Introduction/Motivation =====================================
\section{Introduction}

Simulation based on first principles is an important practice, because it is the only way that a theoretical model can be checked against experiments or real-world data.
In high-energy physics (HEP) experiments, a thorough understanding of the properties of known physics forms the basis of any searches that look for new effects.
This can only be achieved by an accurate simulation, which in many cases boils down to performing an integral and sampling from it. 
Often high-dimensional phase space integrals with non-trivial correlations between dimensions are required in important theory calculations. 
Monte-Carlo (MC) methods still remain as the most important techniques for solving high-dimensional problems across many fields, including for instance: biology~\cite{Hobolth2008, MC.applications}, chemistry~\cite{1955JChPh..23..356R}, astronomy~\cite{MacGillivray1982}, medical physics~\cite{2006PMB....51R.287R}, finance~\cite{jackel2002monte} and image rendering~\cite{kalos2008monte}.
In high-energy physics, all analyses at the Large Hadron Collider (LHC) rely strongly on multipurpose Monte Carlo event generators~\cite{Webber:1986mc,Buckley:2011ms} for signal or background prediction. However, the extraordinary performance of the experiments requires an amount of simulated data that soon cannot be delivered with current algorithms and computational resources~\cite{Buckley:2019wov,ATLASCS}. 

A main endeavour in the field of MC methods is to improve the error estimate. In particular, stratified sampling --- dividing the integration domain in sub-domains, and importance sampling --- sampling from non-uniform distributions~\cite{James:1968gu} are two ways of reducing the variance.
Currently, the most widely used numerical algorithm that exploits importance sampling is the \veg\ algorithm~\cite{Lepage:1977sw,Lepage:1980dq}. But \veg\ assumes the factorizability of the integrand, which can be a bad approximation if the variables have complex correlations amongst one another.
\foam~\cite{Jadach:2002kn} is a popular alternative that tries to address this issue. It uses an adaptive strategy to attempt to model correlations, but requires exponentially large samples in high dimensions.

Lately, the burgeoning field of machine learning (ML) has brought new techniques into the game. For the following discussion, we restrict ourselves to focus on progress made in the field of high-energy physics, see~\cite{Bourilkov:2019yoi} for a recent review. However, these techniques are also widely applied in other areas of research. Concerning event generation,~\cite{Bendavid:2017zhk} used boosted decision trees and generative adversarial networks (GANs) to improve MC integration. Reference~\cite{Klimek:2018mza} proposed a novel idea that uses a dense neural network (DNN) to learn the phase space directly and shows promising results.
In principle, once the neural network(NN)-based algorithm for MC integration is trained, one can invert the network and use it for sampling.
However, the inversion of the NN requires evaluating its Jacobian, which incurs a computational cost that scales as $\mathcal{O}(D^{3})$ for $D$-dimensional integrals~\footnote{An $N$ particle final state phase space is a $D\approx 4N-3$ dimensional integral, when including recursive multichannel selection in the integral.}. 
Therefore, it is extremely inefficient to use a standard NN-based algorithm for sampling.

In addition to generating events from scratch, it is possible to generate additional events from a set of precomputed events. References~\cite{Otten:2019hhl,Hashemi:2019fkn,
  DiSipio:2019imz,Butter:2019cae,Carrazza:2019cnt,SHiP:2019gcl,Butter:2019eyo} used GANs and Variational Autoencoders (VAEs) to achieve this goal. While their work is promising, they have a few downsides. The major advantage of this approach is the drastic speed improvement over standard techniques. They report improvements in generation of a factor around 1000. However, this approach requires a significant number of events already generated which may be cost prohibitive for interesting, high-multiplicity problems. Furthermore, these approaches can only generate events similar to those already generated. Therefore, this would not improve the corners of distributions~\cite{Matchev:2020tbw} and can even result in incorrect total cross-sections. Yet another approach to speed up event generation is to use NN as interpolator and learn the Matrix Element~\cite{Bishara:2019iwh}. 

Our goal is to explore NN architectures that allow both efficient MC integration and sampling.
A ML algorithm based on \emph{normalizing flows} (NF) provides such a candidate.
The idea was first proposed by
\emph{non-linear independent components estimation} (NICE)~\cite{DBLP:journals/corr/DinhKB14,Dinh2016DensityEU}, and generalized in~\cite{rezende2015variational,DBLP:journals/corr/abs-1808-03856,2019arXiv190604032D}, for example.
They introduced \emph{coupling layers} (CL) allowing the inclusion of NNs in the construction of a bijective mapping between the target and initial distributions such that the $\mathcal{O}(D^{3})$ evaluation of the Jacobian can be reduced to an analytic expression. This expression can now be evaluated in $\mathcal{O}(D)$ time. These techniques have also been combined with Markov Chain Monte Carlo methods, showing promising results~\cite{2017arXiv170607561S,2017arXiv171109268L,2019arXiv190303704H}.

Our contribution is a complete, openly available implementation of normalizing flows into \tf~\cite{tensorflow2015-whitepaper}, to be used for any high-dimensional integration problem at hand. Our code includes the original proposal of~\cite{DBLP:journals/corr/abs-1808-03856} and the additions of~\cite{2019arXiv190604032D}. We further include various different loss functions, based on the class of $f$-divergences~\cite{2014ISPL...21...10N}. 
The paper is organized in the following way. The basic principles of MC integration and importance sampling are reviewed in Section \ref{sec:MC.int}. 
In Section \ref{sec:Imp.Samp}, we review the concept of normalizing flows and work done on CL-based flow by \cite{DBLP:journals/corr/DinhKB14,Dinh2016DensityEU,DBLP:journals/corr/abs-1808-03856,2019arXiv190604032D}. We investigate the minimum number of CLs required to capture the correlations between every other input dimension.
Section \ref{sec:setup} sets up the stage for a comparison between our code, \veg, and \foam\ on various trial functions, of which we give results in Section \ref{sec:examples}. This comparison is based on several criteria, allowing a potential user to judge whether it might be worth trying out. Section \ref{sec:conclusions} contains our conclusion and outlook.

\section{Monte Carlo Integrators}
\label{sec:MC.int}
While techniques exist for accurate one-dimensional integration, such as double exponential integration~\cite{takahasi1974double}, using them for high dimensional integrals requires repeated evaluation of one dimensional integrals. This leads to an exponential growth in computation time as a function of the number of dimensions. This is often referred to as the \emph{curse of dimensionality}.
In other words, when the dimensionality of the integration domain increases, the points become more and more sparse and no statistically significant statement can be made without increasing the number of points exponentially. This can be seen in the ratio of the volume of a $D$-dimensional hypersphere to the $D$-dimensional hypercube, which vanishes as $D$ goes to infinity. However, Monte-Carlo techniques are statistical in nature and thus always converge as $1/\sqrt{N}$ for any number of dimensions. 

Therefore, MC integration is the most important technique in solving high-dimensional integrals numerically. 
The na\"ive MC approach samples uniformly on the integration domain ($\Omega$). Given $N$ uniform samples, the integral of $f\left(x\right)$ can be approximated by,
\begin{equation}\label{eq:IS.0}
  I\approx\frac{V}{N}\sum_{i=1}^N f(x_i)\equiv V\,\langle f\rangle_{x}\, ,
\end{equation}
and the uncertainty is determined by the standard deviation of the mean,
\begin{equation}  \label{eq:IS.1}
  \sigma_I=\sqrt{\text{Var}}\approx V\,\sqrt{\frac{\langle {f}^2\rangle_{x}-\langle f\rangle_{x}^2}{N-1}}\;,
\end{equation}
where $V$ is the volume encompassed by $\Omega$ and $\langle\ \rangle_{x}$ indicates that the average is taken with respect to a uniform distribution in $x$.
While this works for simple or low-dimensional problems, it soon becomes inefficient for high-dimensional problems. 
This is what our work is concerned with. In particular, we are going to focus on improving current methods for MC integration that are based on importance sampling.

In importance sampling, instead of sampling from an uniform distribution, one samples from a distribution $g(x)$ that ideally has the same shape as the integrand $f(x)$.
Using the transformation ${\rm d}x={\rm d}G(x)/g(x)$, with $G(x)$ the cumulative distribution function of $g(x)$, 
one obtains
\begin{equation}
  \label{eq:MC.3}
  I=\int_{\Omega} \frac{f(x)}{g(x)}\, \text{d}G(x)=V\,\langle f/g\rangle_{G}\;,
  \qquad
  \sigma_I=V\,\sqrt{\frac{\langle (f/g)^2\rangle_{G}-\langle f/g\rangle_{G}^2}{N-1}}\;.
\end{equation}
In the ideal case when $g(x)\to f(x)/I$, Eq.~\eqref{eq:MC.3} would be estimated with vanishing uncertainty.
However, this requires already knowing the analytic solution to the integral!
The goal is thus to find a distribution $g(x)$ that resembles the shape of $f(x)$ most closely, while being integrable and invertible such as to allow for fast sampling. 
We review the current MC integrators that are widely used, especially in the field of high-energy physics.

\veg~\cite{Lepage:1977sw,Lepage:1980dq}
approximates all 1-dimensional projections of the integrand using a histogram and an adaptive algorithm. This algorithm adjusts the bin widths such that the area of the bins are roughly equal. To sample a random point from \veg\ can be done in two steps. First, select a bin randomly for each dimension. Second, sample a point from each bin according to a uniform distribution. However, this algorithm is limited because it assumes that the integrand factorizes, i.e.
\begin{equation}
f\left(\vec{x}\right) = f_1\left(x_1\right)\cdots f_D\left(x_D\right),
\end{equation}
where $f\colon\mathbb{R}^D\mapsto \mathbb{R}$ and $f_i\colon\mathbb{R}\mapsto \mathbb{R}$. High-dimensional integrals with non-trivial correlations between integration variables, that are often needed for LHC data analyses, cannot be integrated efficiently with the \veg\ algorithm (c.f.~\cite{Hoeche:2019rti}). The resulting uncertainty can be reduced further by applying stratified sampling, in addition to the \veg\ algorithm, after the binning~\cite{VEGAS}.

\foam~\cite{Jadach:2002kn}
uses a cellular approximation of the integrand and is therefore able to learn correlations between the variables. In the first phase of the algorithm, the so-called exploration phase, the cell grid is built by subsequent binary splits of existing cells. Since the first cell consists of the full integration domain, all regions of the integration space are explored by construction. The second phase of the algorithm uses this grid to generate points either for importance sampling or as an event generator. In this work we use the implementation of \cite{Hoeche:foam}, which implemented an additional reweighting of the cells at the end of the optimization.

However, both \foam\ and \veg\ are based on histograms, whose edge effects would be detrimental to numerical analyses that demand high precision. As we will explain below, our code \iflow\  uses a spline approximation which does not suffer from these effects. These edge effects are an important source of uncertainty for high-precision physics~\cite{Dulat:2017prg}. 

%====================================================
\section{Importance Sampling with Normalizing Flows}
\label{sec:Imp.Samp}
As we detailed in the previous section, importance sampling requires finding an approximation $g(x)$ that can easily be integrated and subsequently inverted, so that we can use it for sampling. Mathematically, this corresponds to a coordinate transformation with an inverse Jacobian determinant that is given by $g(x)$.
General ML algorithms incorporate NNs in learning the transformation, which inevitably involve evaluating the Jacobian of the NNs. This results in inefficient sampling. Coupling Layer-based Normailzing Flow algorithms precisely circumvent this problem.
To begin, let us review the concept of a normalizing flow (NF).

Let $c_k$, with $k=1,...,K$, be a series of bijective mappings on the random variable $\vec{x}$:
\begin{equation}
  \vec{x}_K=\;c_K(c_{K-1}(\cdots c_2(c_1(\vec{x})))).
\end{equation}
Based on the chain rule, the output $\vec{x}_K$'s probability distribution, $g_K$, can be inferred given the base probability distribution $g_0$ from which $\vec{x}$ is drawn: 
\begin{equation}\label{eq:nf}
  {g}_K(\vec{x}_K)=\;{g}_0(\vec{x}_0)\prod_{k=1}^{K}\,\left|
  \frac{\partial c_k(\vec{x}_{k-1})}{\partial \vec{x}_{k-1}}\right|^{-1}\;,
  \quad\text{where}\quad
  \begin{cases}
  \vec{x}_0=\vec{x} \\
  \vec{x}_k=c_k(\vec{x}_{k-1})
  \end{cases}\;.
\end{equation}
One sees that the target and base distributions are related by the inverse Jacobian determinant of the transformation.
For practical uses, the Jacobian determinant must be easy to compute, restricting the allowed functional forms of $c_k$. 
However, with the help of \emph{coupling layers},
first proposed by~\cite{DBLP:journals/corr/DinhKB14,Dinh2016DensityEU}, then generalized by~\cite{DBLP:journals/corr/abs-1808-03856,2019arXiv190604032D}, one can incorporate NNs into the construction of $c_k$, thus greatly enhancing the level of complexity and expressiveness of NF without introducing any intractable Jacobian computations.

\begin{figure}[t]
  \centerline{\includegraphics[width=0.75\textwidth]{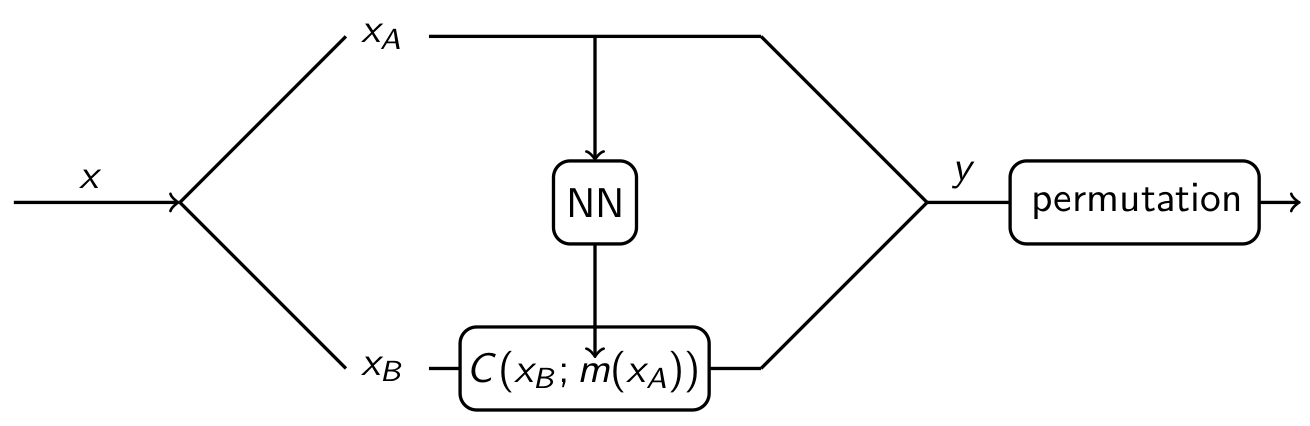}}
  \caption{Structure of a Coupling Layer.
    $m$ is the output of a neural network and defines the
    Coupling Transform, $C$, that will be applied to $x_B$. See Eqs.~\eqref{eq:bijection},~\eqref{eq:bijection2},~\eqref{eq:jacobian_cl} for the mathematical description of a Coupling Layer.}
  \label{fig:CL}
\end{figure}
Figure~\ref{fig:CL}
shows the basic structure of a coupling layer, which is a special design of the bijective mapping $c$. For each map, the input
variable $\vec{x}=\{x_1,..,x_D\}$ is partitioned into two subsets,
$\vec{x}_A$ and $\vec{x}_B$ which can be determined arbitrarily so long as neither is the empty set. This arbitrary partitioning will be referred to as a \emph{masking}. Without loss of generality, one simple partitioning is given by $\vec{x}_A=\{x_1,..,x_d\}$ and $\vec{x}_B=\{x_{d+1},..,x_D\}$. Different maskings can be achieved via permutations of the simple example above.
Under the bijective map, $C$, the resulting variable transforms as
\begin{equation}\label{eq:bijection}
  \begin{split}
    x_A' &= x_A, \quad\quad\quad \quad\quad\quad A\in[1,d]\;,\\
    x_B' &= C(x_B;m(\vec{x}_A)),\quad B\in[d+1,D]\;.
  \end{split}
\end{equation}
The NN takes $x_A$ as inputs and outputs $m(\vec{x}_A)$ that represents the parameters of the invertible
``Coupling Transform'', $C$, that will be applied to $x_B$. We detail various choices for $C$, like piecewise linear, piecewise quadratic, or piecewise rational quadratic spline functions in Appendix~\ref{sec:app:CLdetails}.
The inverse map is given by
\begin{equation}\label{eq:bijection2}
  \begin{split}
    x_A &= x_A'\;,\\
    x_B &= C^{-1}(x_B';m(\vec{x}_A'))=C^{-1}(x_B';m(\vec{x}_A))\;,
  \end{split}
\end{equation}
which leads to the simple Jacobian
\begin{equation}\label{eq:jacobian_cl}
  \left|\frac{\partial c(\vec{x})}{\partial \vec{x}}\right|^{-1}=
  \left|\left(\begin{array}{cc}
    \vec{1} & 0  \\
    \frac{\partial C}{\partial m} \frac{\partial m}{\partial \vec{x}_A}
    & \frac{\partial C}{\partial \vec{x}_B}
  \end{array}\right)\right|^{-1}
  =\left|\frac{\partial C(\vec{x}_B;m(\vec{x}_A))}{\partial \vec{x}_B}\right|^{-1}\;.
\end{equation}
Note that Eq.~\eqref{eq:jacobian_cl} does not require the computation
of the gradient of $m(\vec{x}_A)$, which would scale as $\mathcal{O}(D^3)$
with $D$ the number of dimensions. In addition, taking $\partial C/\partial\vec{x}_B$ to be diagonal further reduces the computation complexity of the determinant to be linear with respect to the dimensionality of the problem. Linear scaling makes this approach tractable even for high dimensional problems. In summary, the NN learns the parameters of a transformation and not the transformation itself, thus the Jacobian can be calculated analytically.

To construct a complete Normalizing Flow, one simply compounds a series of Coupling Layers with the freedom of choosing any of the outputs of the previous layer to be transformed in the subsequent layer. 
We show in Sec~\ref{sec:clnumber} that $2\lceil\log_2\mathrm{D}\rceil$
number of Coupling Layers are required in order to express all non-separable structures of the integrand.

\subsection{Number of Coupling Layers}\label{sec:clnumber}

The minimum number of coupling layers required to capture all possible correlations between every dimension of the integration variable, $n_\text{min}$, depends on the dimensionality of the integral, $D$~\cite{DBLP:journals/corr/abs-1808-03856}. 
In the
cases of $D=2$ and $D=3$, each dimension is transformed once based on the other dimension(s) and thus $n_\text{min}=2$ and $n_\text{min}=3$, respectively. This way of counting  $n_\text{min}$ could be generalized to higher $D$. In fact, this is what \emph{autoregressive flows} are based on~\cite{2017arXiv170507057P}. Here we show that the number of coupling layers required to capture all the correlations is $2\lceil\log_2 D\rceil$ for $D>5$, and $D$ layers for $D\leq5$. This can be considered the minimum number of layers required in order to capture all correlations, adding an additional layer will not add any new information, and similar effects should be achieved with increasing the depth of the network associated with each layer. On the other hand, this can be considered the maximum number of layers needed to capture all the correlations. If a function has fewer correlations, then all the correlations can be captured with less than $2\lceil\log_2 D\rceil$.

\begin{theorem}
Given a set of correlated random variables $x$, if a transformation exists that takes the variables $x$ to $z$, such that the correlation between the variables $z$ is zero, then a composition of normalizing flows can create such a transformation. Given a set of infinitely wide NNs that are universal function approximators, and requiring that all variables are transformed equal number of times, it is possible to represent all the correlations between variables in a normalizing flow using $2\lceil\log_2 D\rceil$ layers for $D > 5$. When $D \leq 5$ it is possible to represent all correlations with $D$ layers.
\end{theorem}

\begin{proof}
Given the random variables $x_1, \ldots x_D$, with means $\mu_1, \ldots, \mu_D$ and joint probability distribution $f\left(x_1, \ldots, x_D\right)$, the correlation between all the variables is given by: 
\begin{equation}
    \langle \left(x_1 - \mu_1\right) \ldots \left(x_D - \mu_D\right)\rangle = \int_0^{1} dx_1 \ldots dx_D \left(x_1 - \mu_1) \ldots (x_D - \mu_D\right) f\left(x_1,\ldots,x_D\right).
    \label{eq:theorem1}
\end{equation}
Using two layers of a normalizing flow network, which can be seen as a universal function approximator,
defines a transformation $T: x \mapsto y$, with the bounds of integration being mapped such that 
$y_i\left(T\left(x=0\right)\right)=0$ and 
$y_i\left(T\left(x=1\right)\right)=1\ \forall i \in [1,D]$, with the sets $\{y_a\} = \{y_i |\, i \equiv 1 \pmod{2},\, i \in [1, D]\}$ and $\{y_b\} = \{y_i |\, i \equiv 0 \pmod{2},\, i \in [1, D]\}$, such that $f\left(x_1, 
\ldots x_D\right) \mapsto g\left(\{y_{a}\}\right)h\left(\{y_{b}\}\right)$, and with 
the Jacobian $J(y,x)$. The transformation also
maps the means: $\mu \mapsto \mu^{y}$. This decomposition is possible following from the arguments of~\cite[Sec. 2.2]{papamakarios2019normalizing}. Applying the 
transformation to Eq.~\ref{eq:theorem1} gives:
\begin{equation}
    \langle \left(x_1 - \mu_1\right) \ldots \left(x_D - \mu_D\right)\rangle 
    = \int_0^{1} dy_1 \ldots dy_D J(y, x)\left(y_1 - \mu^{y}_1) \ldots (y_D - \mu^{y}_D\right) 
    g\left(\{y_{a}\}\right)h\left(\{y_{b}\}\right).
    \label{eq:theorem2}
\end{equation}
If we now consider the correlation between the variables $y$, we obtain:
\begin{align}
    \langle \left(y_1 - \mu^{y}_1\right) \ldots \left(y_D - \mu^{y}_D\right)\rangle =&
    \int_0^{1} dy_1 \ldots dy_D \left(y_1 - \mu^{y}_1\left) \ldots \right(y_D - \mu^{y}_D\right) 
    g\left(\{y_{a}\}\right)h\left(\{y_{b}\}\right) \nonumber\\
    =& \int_0^{1} \prod_{\{y_a\}} \left(dy_{a} \left(y_a - \mu^{y}_a\right)\right) 
    g\left(\{y_a\}\right)
    \times \int_0^{1} \prod_{\{y_b\}} \left(dy_{b} \left(y_b - \mu^{y}_b\right)\right) 
    h\left(\{y_b\}\right) \nonumber\\
    =& \bigg\langle \prod_{\{y_{a}\}} \bigg(y_a - \mu^{y}_a\bigg)\bigg\rangle
    \bigg\langle\prod_{\{y_{b}\}} \bigg(y_b - \mu^{y}_b\bigg)\bigg\rangle.
    \label{eq:theorem3}
\end{align}
The result of the transformation shows that the variables $\{y_a\}$ are now not correlated 
with $\{y_b\}$. We can construct a subsequent transformation $T': y \mapsto z$, with 
$z_i\left(T'\left(y=0\right)\right)=0$ and 
$z_i\left(T'\left(y=1\right)\right)=1\ \forall i \in [1,D]$, and the sets $\{z_a\} = \{z_i|\, i\equiv 1 \pmod{4},\, i\in [1,D]\}$, $\{z_b\} = \{z_i|\, i\equiv 2 \pmod{4},\, i\in [1,D]\}$, $\{z_c\} = \{z_i|\, i\equiv 3 \pmod{4},\, i\in [1,D]\}$, and $\{z_d\} = \{z_i|\, i\equiv 0 \pmod{4},\, i\in [1,D]\}$, such that 
$g\left(\{y_a\}\right)h\left(\{y_{b}\}\right) 
\mapsto g'\left(\{z_a,z_b\}\right)h'\left(\{z_c,z_d\}\right)$ with the constraint 
that such a transformation does not introduce new correlations between the variables that have 
already been decorrelated. In other words, the composition of $T$ and $T'$ can be defined as a
transformation $T'': x \mapsto z$, with $z_i(T''(x)=0)=0$ and $z_i(T''(x)=1)=1\ \forall i \in [1,D]$,
such that:
\begin{equation*}
    f\left(x_1,\ldots,x_N\right) \mapsto
    g_1\left(\{z_a\}\right)g_2\left(\{z_b\}\right)
    g_3\left(\{z_c\}\right)
    g_4\left(\{z_{d}\}\right),
\end{equation*}
and the means are mapped from $\mu$ to $\mu^{z}$.
Thus, the correlation between the variables $z$ is given by:
\begin{align}
    \langle \left(z_1 - \mu^{z}_1\right) \ldots \left(z_D - \mu^{z}_D\right)\rangle 
    =& \int_0^{1} dz_1 \ldots dz_D \left(z_1 - \mu^{z}_1) \ldots (z_D - \mu^{z}_D\right) 
    g_1\left(\{z_a\}\right)g_2\left(\{z_b\}\right)
    g_3\left(\{z_c\}\right)
    g_4\left(\{z_d\}\right) \nonumber\\
    =& \left\langle \prod_{\{z_a\}} \left(z_a-\mu^z_a\right) \right\rangle 
    \left\langle \prod_{\{z_b\}} \left(z_b-\mu^z_b\right) \right\rangle
    \left\langle \prod_{\{z_c\}} \left(z_c-\mu^z_c\right) \right\rangle
    \left\langle \prod_{\{z_d\}} \left(z_d-\mu^z_d\right) \right\rangle.
    \label{eq:theorem3}
\end{align}

The above transformations can be iterated until all the variables are decorrelated. A method of 
determining the mapping for each step can be obtained by the following procedure:
\begin{enumerate}
    \item Reindex the dimension numbers from $[1,D]$ to $[0,D-1]$
    \item Convert all dimensions to their binary representation, using the minimum number of bits 
    required to represent the number $D-1$
    \item Consider the least significant bit for each dimension, and define the transformation as $T: x 
    \mapsto y$, with $f(\{x_0\},\{x_1\}) \mapsto g(\{x_0\})h(\{x_1\})$, where $\{x_0\} (\{x_1\})$ is 
    the set of variables with a 0 (1) for the least significant bit 
    \item Repeat the third step taking the next least significant bit, until the most significant bit is 
    reached
\end{enumerate}

See Table~\ref{tab:masking} for an example of the steps above. In that example, transformation 1 would groups the first 8 dimensions in $g(\{x_0\})$ and the last 4 in $h(\{x_1\})$ etc.

\begin{table}[th]
    \centering
    \begin{tabular}{C{30mm}|C{5mm}|C{5mm}|C{5mm}|C{5mm}|C{5mm}|C{5mm}|C{5mm}|C{5mm}|C{5mm}|C{5mm}|C{5mm}|C{5mm}|C{5mm}|}
         Dimension&0&1&2&3&4&5&6&7&8&9&10&11  \\
         \hline
         Transformation 1& 0 & 1 & 0 & 1 & 0 & 1 & 0 & 1 & 0 & 1 & 0 & 1  \\
         Transformation 2& 0 & 0 & 1 & 1 & 0 & 0 & 1 & 1 & 0 & 0 & 1 & 1  \\
         Transformation 3& 0 & 0 & 0 & 0 & 1 & 1 & 1 & 1 & 0 & 0 & 0 & 0   \\
         Transformation 4& 0 & 0 & 0 & 0 & 0 & 0 & 0 & 0 & 1 & 1 & 1 & 1  \\
    \end{tabular}
    \caption{Finding the unique masking to capture all correlations in a $D=12$ space, using the procedure detailed above.}
    \label{tab:masking}
\end{table}

The number of steps for this procedure can easily be seen to be $\lceil \log_2\left(D\right) \rceil$. 
However, since we need two layers per transformation to ensure that all variables are equally 
transformed by the network leads to a requirement of $2\lceil \log_2\left(D\right) \rceil$.

In the situation of $D \leq 5$, the $i^{\text{th}}$ coupling layer can be defined to take the 
$i^{\text{th}}$ variable and transform the variable, such that it is not correlated with any other 
variable. This leads to a requirement of $D$ layers.
\end{proof}

The requirement on the above theorem is that we require that a transformation exists in order to perform the above mapping. However, even if a transformation does not exist for the integrand itself, with the use of importance sampling, it is only necessary to find a function $g$ which is as close to the integrand as possible. In \iflow\, splines are used to create the function $g$, and according to the Stone-Weierstrass Theorem~\cite{PINKUS20001, 10.2307/1989788, 10.2307/3029337}, it is possible to represent $g$ such that it is $\varepsilon$ close to $f$. A corollary of the Stone-Weierstrass Theorem for $\mathbb{R}^n$ can be expressed as:
\begin{corollary}
Given a function $f: \mathbb{R}^n \mapsto \mathbb{R}$, $\varepsilon > 0$, and $C(\mathbb{R}^n, \mathbb{R})$: the space of all real-valued continuous functions in $\mathbb{R}^n$, there exists a polynomial spline $g \in C(\mathbb{R}^n, \mathbb{R})$ such that:
\begin{equation}
    |f(x) - g(x)| < \varepsilon,
\end{equation}
for all $x \in \mathbb{R}^n$.
\end{corollary}
\begin{proof}
The space $R^n$ is a subset of the spaces proved in the Stone-Weierstrass Theorem, and thus the proof of the corollary follows directly from the Stone-Weierstrass Theorem.
\end{proof}
\noindent Furthermore, this can be extended to a sum of piecewise polynomials, such that any continuous and bounded function $f$ can be represented by an infinite series of polynomials (see Theorem C from~\cite{PINKUS20001}). In \iflow, we will consider the case of discontinuous functions, but these can be approximated as a continuous function with a slope of $1/\varepsilon$ in the region of discontinuity. This will lead to some difference between $f$ and $g$, but since the goal is to find a function $g$ as close to $f$ as possible, then this is acceptable and should still allow for high precision importance sampling.

\subsection{Using \iflow}
\begin{figure}[ht]
  \centering
  \includegraphics[width=0.75\textwidth,trim = 100 350 150 200, clip=true]{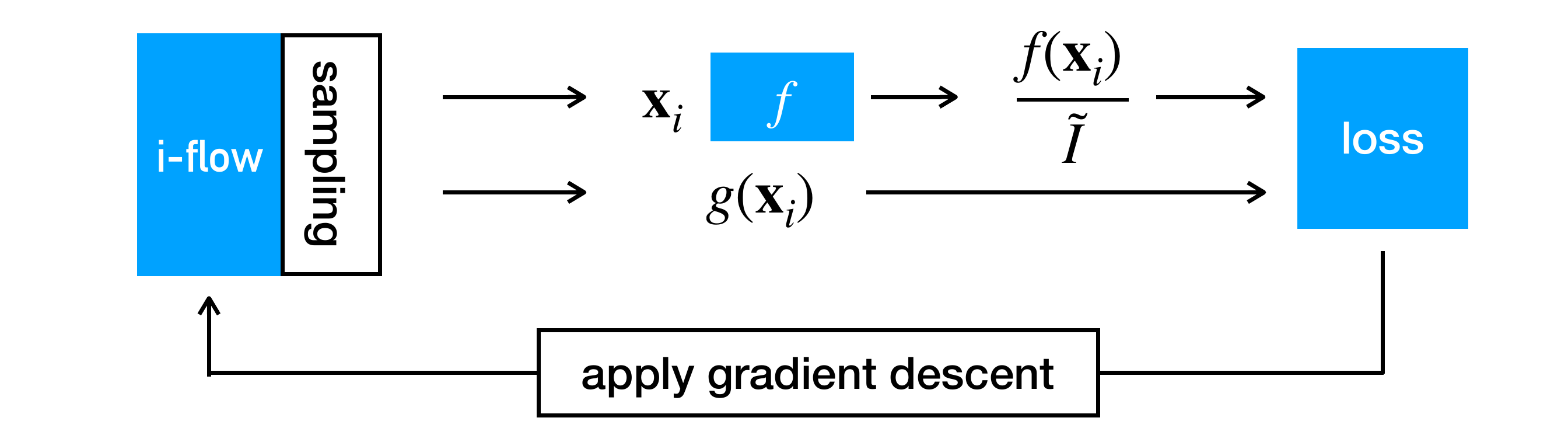}
  \caption{Illustration of one step in the training of \iflow.
    Users need to provide a normalizing flow network,
    a function $f$ to integrate, and a loss function.
    $\tilde{I}$ stands for the Monte-Carlo estimate
    of the integral using the sample of points $\vec{x}_i$, and $g(\vec{x}_i)$ is the probability of a given point occurring in the \iflow\ sampling.}
  \label{fig:NFonestep}
\end{figure}
The \iflow\ package requires three pieces of information from the user: the function to be integrated, the normalizing flow network, and the method of optimizing the network.
Figure~\ref{fig:NFonestep} shows schematically how one step in the training of \iflow\ works. The code is publicly available on gitlab at {\tt https://gitlab.com/i-flow/i-flow}. Running the script {\tt iflow\_test.py} will produce the results presented in section~\ref{sec:examples}.

\subsubsection{Integrand}

The function to be integrated has very few requirements on how it is implemented in the code. Firstly, the function must accept an array of points with shape $(n_{\text{batch}}, D)$, where $n_{\text{batch}}$ is the number of points to sample per training step. Secondly, the function must return an array with shape $(n_{\text{batch}})$ to be used to estimate the integral. Finally, the number of dimensions in the integral is required to be at least 2. However, one dimension can be treated as a dummy dimension integrated from 0 to 1, which will not change any result.

\subsubsection{Normalizing flow network}

A normalizing flow network consists of a series of coupling layers compounded together. To construct each coupling layer, one needs to specify the choice of coupling transform $C$ (cf App.~\ref{sec:app:CLdetails}), the number of coupling layers, the masking for each level, and the neural network $m(x_A)$ that constitutes the coupling transform. We provide the ability to automatically generate the masking and number of layers according to Sec.~\ref{sec:clnumber}.

The neural networks $m(x_A)$ must be provided by the user. However, we provide examples for a dense network and the U-shape network of~\cite{DBLP:journals/corr/abs-1808-03856}. This provides the user the flexibility to achieve the expressiveness required for their specific problem.

\subsubsection{Optimizing the network}

To uniquely define the optimization algorithm of the network, two pieces of information are required. 
Firstly, the loss function to be minimized is required. We supply a large set of loss functions from the set of $f$-divergences, which can be found in App.~\ref{sec:app.loss}. By default, the \iflow\ code uses the exponential loss function. 
Secondly, an optimizer needs to be supplied. In the examples we used the ADAM optimizer~\cite{kingma2014adam}. However, the code can use any optimizer implemented within \tf.
\subsubsection{Hyperparameters}
The setup we presented here has several hyperparameters that can be adjusted for better performance. However, \iflow\ has the flexibility for the user to implement additional features in each section beyond what is discussed below. This would come with additional hyperparameters as well.

The first group concerns the architecture of the NNs $m(x_A)$. Once the general type of network (dense or U-shape) is set, the number of layers and nodes per layer have to be specified. In the case of the U-shape network, the user can specify the number of nodes in the first layer and the number of ``downward" steps. 

The second group of hyperparameters concerns the optimization process. Apart from setting an optimizer ({\it e.g.} ADAM~\cite{kingma2014adam}), a learning schedule ({\it e.g.} constant or exponentially decaying), an initial learning rate, and a loss function have to be specified. Some of these options come with their own, additional set of hyperparameters. The number of points per training epoch and the number of epochs have to be set as well. 

The third group of hyperparameters concerns the setup of \iflow\ directly. As was discussed in~\cite{DBLP:journals/corr/abs-1808-03856}, there are two ways to pass $x_A$ into $m(x_A)$: either directly or with one-blob encoding. \iflow\ supports both of these options. One-blob encoding~\cite{DBLP:journals/corr/abs-1808-03856} is a generalization of one-hot encoding. The input $x_A$ is passed through a Gaussian kernel and several adjacent bins are activated. If one-blob encoding is used, the number of input bins has to be specified, the width of the Gaussian is set to the inverse of the number of bins. Further, the type of coupling function $C(x_B,m(x_A))$, the number of output bins, the number of CLs and the maskings have to be set.

\subsubsection{Putting it all together}
The networks are trained by sampling a fixed number of points using the current state of $g(x)$~\footnote{Since we initialize the last layer of each network with vanishing bias and weights, in the first sampling $g(x)$ is constant.}. We use one of the statistical divergences as a measure for how much the distribution $g(x)$ resembles the shape of the integrand $f(x)$, and an optimizer to minimize it. Because we can generate an infinite set of random numbers and evaluate the target function for each of the points, this approach corresponds to supervised learning with an infinite dataset. Drawing a new set of points at every training epoch automatically also ensures that the networks cannot overfit.

%= Integrator Comparison ==================================================

\section{Integrator Comparison}
\label{sec:setup}
To illustrate the performance of \iflow\ and compare it to \veg\ and \foam, we present a set of six test functions, each highlighting a different aspect of high-dimensional integration and sampling. These functions demonstrate how each algorithm handles the cases of a purely separable function, functions with correlations, and functions with non-factorizing hard cuts. In most cases, an analytic solution to the integral is known.

The first test function is an $n$-dimensional Gaussian, serving as a sanity check:
\begin{equation}
    \label{eq:test.1}
    f_1(\vec{x}) =  (\alpha  \sqrt\pi)^{-n} \exp{\{-\frac{\sum_i(x_i-0.5)^2}{\alpha^2}\}}\;.
\end{equation}
The result of integrating $f_1$ from zero to one is given by:

\begin{equation}
    \int_0^1 d^n \vec{x}\, f_1(\vec{x}) = \text{erf}\left(\frac{1}{2\alpha}\right)^n\;.
\end{equation}
In the following, we use $\alpha = 0.2$. 

The second test function is an $n$-dimensional Camel function, which would show how \iflow\ learns correlations that \veg\ (without stratified sampling) would not learn:   
\begin{equation}
    \label{eq:test.2}
    f_2(\vec{x}) =  \frac{1}{2}(\alpha  \sqrt\pi)^{-n} \left(
    \exp{\{-\frac{\sum_i(x_i-\tfrac{1}{3})^2}{\alpha^2}\}}+\exp{\{-\frac{\sum_i(x_i-\tfrac{2}{3})^2}{\alpha^2}\}}\right)\;.
\end{equation}
The result of integrating $f_2$ from zero to one is given by:
\begin{equation}
    \int_0^1 d^n \vec{x}\, f_2(\vec{x}) = \left(\frac{1}{2}\left(\text{erf}\left(\frac{1}{3\alpha}\right) + \text{erf}\left(\frac{2}{3\alpha}\right)\right)\right)^n\;.
\end{equation}
In the following, we use $\alpha = 0.2$. 

The third case is given by
\begin{align}
\begin{aligned}
    \label{eq:test.3}
            %dx1, dy1, rr, w1, ee = 0.4, 0.6, 0.25, 1./0.004, 3.0
        f_3(x_1,x_2)&=x_2^a \exp{\{-w|(x_2-p_2)^2+(x_1-p_1)^2-r^2|\}}\\
        &+(1-x_2)^a  \exp{\{-w|(x_2-1+p_2)^2+(x_1-1+p_1)^2-r^2|\}}\;.
\end{aligned}
\end{align}
This function has two circles with shifted centers, varying thickness and height. Also, the function exhibits non-factorizing behavior. The integral of $f_3$ between 0 and 1 can be computed numerically using Mathematica~\cite{Mathematica}, which is $0.0136848\pm (5\cdot 10^{-9})$,
with $p_1 = 0.4, p_2 = 0.6, r=0.25, w = 1/0.004$ and $a = 3$~\footnote{There is no known analytic solution to this given function}.

The fourth case is an annulus function with hard cuts: 
\begin{equation}
    \label{eq:test.4}
    f_{4}(x_1,x_2) = \left\{ \begin{array}{cc}
        1 &  0.2 < \sqrt{x_1^2 + x_2^2}<0.45\\
        0 & \text{else}
    \end{array}\right.\;.
\end{equation}
This function demonstrates how \iflow\ learns hard, non-factorizing cuts. 
The result of integrating $f_4$ from zero to one is given by:
$\pi\left(0.45^2-0.2^2\right) = 0.1625 \pi$.

The fifth case is motivated by high energy physics, and is a one-loop scalar box integral representative of an integral required for the calculation of $gg\to gh$ in the Standard Model. This calculation is an important contribution for the total production cross-section of the Higgs boson. As explained in App.~\ref{sec:app.box}, after Feynman parametrisation and sector decomposition~\cite{Binoth:2002xh}, the integral of interest is given by
\begin{equation}
\label{eq:test.5}
\begin{split}
    f_5=&S_{Box}(s_{12},s_{23},s_1,s_2,s_3,s_4,m_t^2,m_t^2,m_t^2,m_t^2)\\
    &+S_{Box}(s_{23},s_{12},s_2,s_3,s_4,s_1,m_t^2,m_t^2,m_t^2,m_t^2)\\
    &+S_{Box}(s_{12},s_{23},s_3,s_4,s_1,s_2,m_t^2,m_t^2,m_t^2,m_t^2)\\
    &+S_{Box}(s_{23},s_{12},s_4,s_1,s_2,s_3,m_t^2,m_t^2,m_t^2,m_t^2)\\
    S_{Box}&(s_{12},s_{23},s_1,s_2,s_3,s_4,m_1^2,m_2^2,m_3^2,m_4^2)=\int_0^1 dt_1dt_2dt_3 \frac1{\tilde{\mathcal{F}}_{Box}^2}\\
    \tilde{\mathcal{F}}_{Box}&=(-s_{12})t_2+(-s_{23})t_1t_3+(-s_1)t_1+(-s_{2})t_1t_2+(-s_{3})t_2t_3\\
    &\quad\quad+(-s_{4})t_3+(1+t_1+t_2+t_3)(t_1m_1^2+t_2m_2^2+t_3m_3^2+m_4^2)
\end{split}
\end{equation}
The result of integrating $f_5$ from zero to one can be obtained through the use of LoopTools~\cite{HAHN1999153}, which gives a numerical result of $1.9369640238\cdot 10^{-10}$ for the inputs $s_{12} = 130^2, s_{23} = -130^2, s_1 = 0, s_2 = 0, s_3 = 0, s_4 = 125^2, m_t = 175$.

As a sixth test function, we consider the polynomial 
\begin{equation}
    \label{eq:test.6}
    f_{6}(x_1,\dots,x_n) = \sum_{i=1}^n -x_i^2 + x_i.
\end{equation}
The result of integrating $f_6$ from zero to one is given by:
\begin{equation}
    \label{eq:test.7}
    \int_0^1 \, dx_1 \dots dx_n\, f_{6}(x_1,\dots,x_n) = \frac{n}{6}
\end{equation}
This function can easily be integrated in a high number of dimensions and, unlike the Gaussian or Camel functions, has support in almost all of the integration domain. It therefore does not suffer that much from the curse of dimensionality.

Further applications to event generation of high-energy particle collisions is discussed in~\cite{Gao:2020zvv} and also in~\cite{Bothmann:2020ywa}. These papers investigate using normalizing flows to improve upon phase space integration for event simulation at particle colliders. The integral dimension for processes with $n_f$ particles in the final state is $D=4n_f-3$. In~\cite{Gao:2020zvv}, we studied processes with $n_f \leq 6$. 

\section{Results}
\label{sec:examples}

\begin{table}[]
\begin{center}
\begin{tabular}{ |l c |  c| c| c |}
\hline
  &Dim& \veg & \foam & \iflow  \\ 
  \hline
 \multirow{4}{4em}{Gaussian} &2& $\mathbf{164,436}$ & $6,259,812 $& $2,310,000$  
 \\ 
&4& $\mathbf{631,874}$& $24,094,679$ &  $2,285,000$
\\ 
&8& $\mathbf{1,299,718}$ & $>50,000,000\; \dagger$& $3,095,000$
\\ 
&16& $\mathbf{2,772,216}$ & $>50,000,000\; \dagger$ & $7,230,000$ 
\\ 
\hline
 \multirow{4}{4em}{Camel} &2& $\mathbf{421,475}$ & $5,619,646$ & $2,225,000$
 \\ 
&4 & $24,139,889$ & $21,821,075$ & $\mathbf{8,220,000}$
\\ 
&8  &\;$>50,000,000\; \dagger$ \;&\; $>50,000,000$ \;&\; $\mathbf{19,460,000}$ \;
\\ 
&16& $993,294 \; \dagger$  
& $>50,000,000\; \dagger$ & $32,145,000\; \dagger$\cite{fn:camel16}
\\ 
  \hline
Entangled circles &2& $43,367,192$ & $\mathbf{17,499,823}$ & $23,105,000$
\\
  \hline
Annulus w. cuts &2& $4,981,080\; \dagger$ & $\mathbf{11,219,498}$ & $17,435,000$
\\
  \hline
Scalar-top-loop &3 & $\mathbf{152,957}$ & $5,290,142$& $685,000$
\\
    \hline
    \hline
 \multirow{3}{4em}{Polynomial} & 18 &$42,756,678$&$>50,000,000$& $\mathbf{585,000}$\\
 & 54 &$>50,000,000$&$>21,505,000 \;*$&$\mathbf{685,000}$\\
 & 96 &$>50,000,000\; \dagger$&$>10,235,000 \;*$&$\mathbf{1,145,000}$\\
  \hline
\end{tabular}
\end{center}
\caption{Number of functional calls to reach a total relative uncertainty of $10^{-4}$ (for the first 11 cases) or $10^{-5}$ (for the last 3 cases). The total relative uncertainty is defined as the inverse-variance weighted combination of the uncertainties of each optimization iteration divided by the true integral value.  The integrator with the fewest functional calls, which also is within 5 standard deviations of the true result, is highlighted in boldface. We set an upper cut-off of $5\cdot 10^7$ calls. A $\dagger$ indicates that the algorithm did not converge to the true integral value within 5 standard deviations (see tab.~\ref{tab:results2}), a $*$ indicates cases where the algorithm ran out of memory before the cut-off was reached. $\alpha=0.2$ for Gaussian and Camel functions.}
\label{tab:results}
\end{table}

\begin{table}[]
\begin{center}
\begin{tabular}{ |l c |  c r| c r| c r| c|}
\hline
  &Dim& \veg & (pull) & \foam & (pull)& \iflow & (pull) & true value\\ 
  \hline
 \multirow{4}{4em}{Gaussian} &2& $0.99925(10)$ &$0.7$&$0.99925(10)$ & $0.6$ &$\mathbf{0.99919(10)}$ &$\mathbf{0.1}$ & $0.999186$
 \\ 
&4& $0.99861(10)$& $2.4$& $\mathbf{0.99835(10)}$ &$\mathbf{-0.2}$& $0.99841(10)$ &$0.4$& $0.998373$
\\ 
&8& $0.99694(10)$ &$1.9$& $ 0.99439(37)\; \dagger$&$-6.4$& $\mathbf{0.99684(10)}$&$\mathbf{0.9}$& $0.996749$
\\ 
&16& $0.99357(10)$ &$0.6$& $0.54986(235) \; \dagger$ &$-188$& $\mathbf{0.99354(10)}$ &$\mathbf{0.4}$& $0.993509$
\\ 
\hline
 \multirow{4}{4em}{Camel} &2& $0.98175(10)$ &$0.9$& $0.98163(10)$ &$-0.3$ & $\mathbf{0.98165(10)}$&$\mathbf{-0.1}$& $0.98166$
 \\ 
&4 & $0.96345(10)$ &$-2.2$& $0.96361(10)$&$-0.5$& $\mathbf{0.96365(10)}$ & $\mathbf{-0.02}$& $0.963657$
\\ 
&8 &$0.92495(28) \; \dagger$&$-13$ &$0.92798(19) \; \dagger$ &$-3.5$& $\mathbf{0.92843(9)}$ &$\mathbf{-2.2}$ & $0.928635$
\\ 
&16& $0.43137(9)$ &$-5001$ & $0.76921(129) \; \dagger$ & $-72$& $\mathbf{0.85940(9)}$\;\cite{fn:camel16}&$\mathbf{-34}$& $0.862363$
\\ 
  \hline
Entangled circles &2& $0.0136798(14)$& $-3.6$ & $\mathbf{0.0136838(14)}$ &$\mathbf{-0.7}$ & $0.0136829(14)$ &$-1.4$& $0.0136848$
\\
  \hline
Annulus w. cuts &2& $0.509813(51)$ &$-14$& $0.510559(51)$ & $1.0$ & $\mathbf{0.510511(51)}$ &$\mathbf{0.1}$ & $0.510508$
\\
  \hline
Scalar-top-loop &3 &\; $1.93711(19)\cdot 10^{-10}$ &$0.7$&\; $\mathbf{1.93708(19)\cdot 10^{-10}}$&$\mathbf{0.6}$&\;$ 1.93677(19)\cdot 10^{-10}$ &$-1.0$&\; $1.936964\cdot 10^{-10}$\;
\\
    \hline
    \hline
\multirow{3}{4em}{Polynomial} & 18 &$2.99989(3)$ &$-3.6$&$2.99986(12) \; \dagger$&$-1.1$&$\mathbf{2.99997(3)}$&$\mathbf{-1.1}$&$3$\\
 & 54 &$8.99972(19) \; \dagger$&$-1.5$&$9.00013(32) \;*$&$0.4$&$\mathbf{9.00001(9)}$&$\mathbf{0.2}$&$9$\\
 & 96 &$0.15547(52) \; \dagger$&$-30683$&$16.0004(3)\; *$&$1.7$&$\mathbf{15.9998(2)}$&$\mathbf{-1.2}$& $16$ \\
  \hline
\end{tabular}
\end{center}
\caption{Integral estimate and uncertainty of the runs of Table~\ref{tab:results} together with their relative deviations (``pull"), defined in Eq.~\eqref{eq:rel.dev}. A $\dagger$ indicates that the algorithm reached a cut-off of $5\cdot 10^7$ function calls before the target uncertainty was reached, a $*$ indicates cases where the algorithm ran out of memory before the cut-off was reached. The result with the smallest relative deviation is boldfaced. $\alpha=0.2$ for Gaussian and Camel functions.}
\label{tab:results2}
\end{table}

\begin{table}[ht]
 \begin{center}
 \begin{tabular}{ |l c |  c| c| c |}
 \hline
   &Dim& \veg & \foam & \iflow
   \\ 
   \hline
  \multirow{4}{4em}{Gaussian} &2& $ \mathbf{7\cdot 10^{-4}}$ &$ 3\cdot 10^{-3}$ & $ 2\cdot 10^{-3}\; *$ 
  \\ 
 &4& $ \mathbf{1.5\cdot 10^{-3}}$&$ 3\cdot 10^{-3}$ &$ \mathbf{1.5\cdot 10^{-3}\; *}$  
 \\ 
 &8& $ 2.5\cdot 10^{-3}$ &$ 3\cdot 10^{-2}$ & $\mathbf{ 1.5\cdot 10^{-3}\; *}$  
 \\ 
 &16&\; $ 3.5\cdot 10^{-3}$ \;&\; $ 2\cdot 10^{-2}$ \;&\;  $ \mathbf{2.5\cdot 10^{-3}\; *}$\;  
 \\ 
 \hline
  \multirow{4}{4em}{Camel} &2& $ \mathbf{2\cdot 10^{-3}}$ & $ \mathbf{2\cdot 10^{-3}}$ & $ \mathbf{2\cdot 10^{-3} \; *}$ 
  \\ 
 &4 & $ 8\cdot 10^{-3}$ & $ 1\cdot 10^{-2}$ & $\mathbf{ 4\cdot 10^{-3}}$ 
 \\ 
 &8 &$ 4\cdot 10^{-2}$ & $ 1.6\cdot 10^{-2}$ & $\mathbf{ 5\cdot 10^{-3}}$ 
 \\ 
 &16& $\dagger$  & $1.5\cdot 10^{-1}$ &  $\mathbf{ 5\cdot 10^{-3}}$
 \\ 
   \hline
 Entangled circles &2& $ 1\cdot 10^{-2}$& $\mathbf{ 4\cdot 10^{-3}}$ & $ 5\cdot 10^{-3} \; *$
 \\
   \hline
 Annulus w. cuts &2& $\mathbf{ 3\cdot 10^{-3}}$ &$4\cdot 10^{-3}\; *$ & $ 5\cdot 10^{-3}$
 \\
   \hline
 Scalar-top-loop &3 & $ 7\cdot 10^{-4}$ & $  \mathbf{5\cdot 10^{-4}}$ & $ \mathbf{5\cdot 10^{-4} \; *}$
 \\
     \hline
    \hline
 \multirow{3}{4em}{Polynomial} & 18 &$ 1.5\cdot 10^{-3}$&$1.5\cdot 10^{-3}\; *$&$\mathbf{ 8\cdot 10^{-5}\; *}$\\
 & 54 &$ 3\cdot 10^{-3}$&$9\cdot 10^{-4}\; *$&$\mathbf{ 8\cdot 10^{-5}\; *}$\\
 & 96 &$\dagger$&$8\cdot 10^{-4}\; *$&$\mathbf{ 1\cdot 10^{-4}\; *}$\\
   \hline
 \end{tabular}
 \end{center}
 \caption{Relative uncertainty on the integral estimate of the last iteration of the runs of Table~\ref{tab:results}, based on a sample of 5000 points. The integrator that adapted best to the integrand is boldfaced. A $*$ indicates when the value was still decreasing and had not yet converged, a $\dagger$ is in place where the algorithm did not converge to the true integrand.}
 \label{tab:results3}
 \end{table}

In this section we show the performance of \iflow\ and compare it to \veg\ and \foam\ based on the test functions we introduced in Sec.~\ref{sec:setup}.   
For the \veg\ algorithm, we use the default parameters as implemented in~\cite{VEGAS}. This includes the use of stratified sampling and a maximum of 1000 bins per axis. We further set the number of points per iteration to 5000. However, the implementation in~\cite{VEGAS} uses this number as a maximum, so we monitor the actual number of function calls separately. The setup of \foam\ requires a number of points per cell, which we fix to 5000. 
In the setup of \iflow, we use $2 \lceil \log_2 D\rceil$ number of coupling layers with the masking discussed in Sec.~\ref{sec:clnumber}, and the coupling transform $C$ taken to be a Piecewise Rational Quadratic spline (Appendix~\ref{sec:app.PRQ}). The neural network in each CL is taken to be a DNN of 5 layers with 32 nodes in each of the first four layers. The number of nodes in the last layer depends on the coupling transform $C$ and the dimensionality of the integrand. For the case of Piecewise Rational Quadratic splines, the number of nodes is given by $d\cdot(3n_{\text{bins}}+1)$, where $d$ is the number of dimensions to be transformed.  We further set the number of bins ($n_{\text{bins}}$) in each output dimension to 16. The learning rate was set to $1\cdot 10^{-3}$ in all cases. We use the exponential divergence, see eq.~\eqref{eq:exp.div}, as loss function. 

To compare the integrators, we set a relative uncertainty on the integral estimate as target. We then optimize the algorithms until the standard deviation of the inverse-variance weighted combination of the estimates of each optimization iteration (epoch) reaches this target. The inverse-variance weighted combination is defined as:
\begin{equation}
    \mu = \frac{\sum_i \mu_i/\sigma_i^2}{\sum_i 1/\sigma_i^2},\quad
    \sigma^2 = \frac{1}{\sum_i 1/\sigma_i^2},
\end{equation}
where $\mu_i (\sigma_i)$ is the mean (standard deviation) of the $i^{\text{th}}$ epoch and $\mu(\sigma)$ is the combination.
The relative uncertainty is defined as the uncertainty of the given estimate, normalized to the true value of the integral \footnote{Note that \foam\ directly gives the uncertainty including all sampled points up to the given iteration and no combination is needed.}. Given this setup, there are three metrics that we use to compare the integrators: 1) the number of function calls needed to reach the target uncertainty; 2) how close the estimated integral value is to the true value; 3) the uncertainty of the estimates in the last iterations. Each of those highlights a different aspect of the integrator and we detail them below. The results are shown in Tables~\ref{tab:results}--\ref{tab:results3}. We chose a relative target uncertainty of $10^{-4}$ for the non-polynomial test functions and $10^{-5}$ for the polynomials. For Gaussian and Camel functions, we use $\alpha=0.2$. In addition, we set a cut-off at $5\cdot 10^7$ function calls. 

{\it Number of function calls.} This number shows how often the integrand was evaluated by the algorithm until the target uncertainty was reached. Having fewer function calls is especially important when the function is numerically expensive to evaluate and the computational overhead of the integration algorithm becomes subleading. The results are shown in Tab.~\ref{tab:results}. We highlight the entry with the fewest calls in boldface. In addition, we mark entries in which the final integral estimate differs by more than 5 standard deviations from the true result with a $\dagger$ and entries in which too much memory was required by a $*$.

{\it Integral estimate and uncertainty.} This obviously shows how well the integrator estimated the value of the integral. We show our results in Tab.~\ref{tab:results2} and compare them to the true, known results. We highlight in boldface the entry with the smallest relative deviation (``pull"), defined as
\begin{equation}
\label{eq:rel.dev}
    \frac{I_\text{code}-I_\text{true}}{\sqrt{(\Delta I_\text{code}^2 + \Delta I_\text{true}^2)}}\;.
\end{equation} 
Here, $I_{\text{code}}$ is the result from \veg, \foam, or \iflow, $I_{\text{true}}$ is the true value of the integral, and the $\Delta I$ terms signify the uncertainty in the integral. Note, that $\Delta I_{\text{true}}$ is only non-zero for the case of the entangle circles for which it is $5\cdot10^{-9}$.
Cases in which the cut-off for function calls was reached (see Tab.~\ref{tab:results}) are marked with a $\dagger$, cases that ran into memory problems are marked with a $*$. 

{\it Relative uncertainty on the integral estimate in the last iterations.} The uncertainty on the integral estimate after adaptation is a measure for how well the algorithm adapted to the integrand. Once the algorithm is fully adapted, the uncertainty of a single integral estimate will be constant and the combination of all iterations will follow the $1/\sqrt{N}$ scaling law for MC estimates based on $N$ points. A better adapted algorithm introduces a smaller coefficient for that scaling and therefore require fewer function calls to reach a smaller uncertainty. We show our results in Tab.~\ref{tab:results3}. Cases in which \veg\ failed to converge to the right integral value are marked with a $\dagger$, a $*$ shows entries that still showed a downward trend at the end of the optimization, indicating that the algorithm was still adapting to the integrand. We highlight the integrator with the smallest uncertainty in boldface.

For the Gaussians, \veg\ always has the fewest calls. This is expected, since the integrand factorizes. However, the number grows rapidly for increasing integrand dimension, whereas the number for \iflow\ grows slower. \foam\ is not able to reach the target uncertainty for $D=8, 16$ before the cut-off of $5\cdot 10^7$ function calls. \iflow\ has adapted best to all Gaussians of $D>2$, as can be seen in Tab.~\ref{tab:results3}. This means that if a sufficiently small target uncertainty is required, \iflow\ would potentially need fewer function calls to reach it. The fact that the optimization of \iflow\ was not complete when the target uncertainty was reached can also be seen in Fig.~\ref{fig:4dgauss}, where the accumulated uncertainty of \iflow\ (red line) was falling quicker than $1/\sqrt{N}$ (dashed gray lines). In almost all of the cases, the integral estimate of \iflow\ was closest to the true integral value. 

For the Camel functions, \veg\ only has the fewest calls for $D=2$, in higher dimensions \iflow\ needs fewer calls. Note that in 16 dimensions, \veg\ completely misses one of the two peaks, yielding an estimate that is off by a factor of two. Since in this case the integrand is like a Gaussian, \veg\ converges quicker than the other algorithms. The integral estimate of \iflow\ seems also off~\cite{fn:camel16}, but this is due to the fact that it needs roughly 200 epochs to ``see" the structure of the integrand and all of those iterations contribute to the final number in Tab.~\ref{tab:results}. Again, \foam\ needs too many points for $D>4$ to reach the target uncertainty, the integral estimates of \iflow\ are closest to the true value, and \iflow\ has adapted best to the integrand. The latter can be seen by the small relative uncertainties in Tab.~\ref{tab:results3} and the scaling in the 4 dimensional case shown in Fig.~\ref{fig:4dcamel}. 

We discuss the entangled circles and the annulus after the polynomials. For the scalar top loop, \veg\ needs the fewest function calls, but all 3 integrators estimate the true value within one standard deviation. It is, however, interesting to see that both \veg\ and \foam\ seem to be fully adapted, whereas the uncertainties of the estimates of \iflow\ were improving much faster than the $1/\sqrt{N}$ expectation, see Fig.~\ref{fig:box}.  

The polynomials show the strength of \iflow. It has no problems adapting to the high-dimensional integrand, as can be seen in Tab.~\ref{tab:results3}. Therefore, \iflow\ needs comparatively few function calls to reach the target uncertainty. Since the polynomial does not factorize, \veg\ does not adapt well, or in the case of $D=96$ not at all. The difference between the adaptation of the algorithms is also visible in Fig.~\ref{fig:poly}. There, however, we see an interesting pattern in the accumulated uncertainty of \foam\ that we want to comment on. First, since \foam\ estimates the integral and uncertainty of a given iteration always on the points of all previous iterations, the uncertainty can grow for a growing number of points if the central value shifts. Second, due to the symmetry of the polynomial integrand, we see a periodic pattern that we can understand as follows. We start with an uncertainty based on the first 5000 points. Adding more points at this initial stage lets the algorithm ``see" more structure of the integrand and the uncertainty grows. A large cell with a large spread of functional values within it is then further split consecutively into many smaller cells. That reduces the spread of functional values per cell and therefore the uncertainty of the integral estimate. Once the uncertainty drops below a certain value, \foam\ stops splitting these (smaller) cells and returns to one of the ``bigger" cells it did not split in the beginning and starts splitting it. This initially increases the uncertainty again, because the spread of functional values in the large cell is larger than it was in the smaller cells. The result is the oscillating pattern we see in Fig.~\ref{fig:poly}. Note that the minima of this pattern follow the $1/\sqrt{N}$ scaling. 

The entangled circles are best integrated by \foam, as it is only 2 dimensional, yet non-factorizable. \iflow\ is slightly worse, but not by much. \veg, however, does not perform well. Similar statements can be made about the annulus function with hard cuts. \veg\ does the worst because of the non-factorizing structure of the integrand and \foam\ does well because it is only a 2-dimensional problem. As discussed in the earlier sections, \iflow\ also allows efficient sampling once it ``learns'' the integral up to small uncertainty, we therefore use these test functions to illustrate the sampling performance of \iflow. As an example, Fig.~\ref{fig:result.4} shows a sample distribution after training \iflow\ with $5\cdot 10^6$ points (1000 epochs with 5000 points per epoch) on the annulus function of Eq.~\eqref{eq:test.4}. For training, we used a learning schedule with exponential decay. An initial learning rate of $2 \cdot 10^{-3}$ is halved every 250 epochs. The cut efficiency, defined as the fraction of the generated points that pass the hard cut, is $89.6\%$. Figure~\ref{fig:result.3} shows the weights of 10000 points sampled after training with 10M points on the Entangled Circles of Eq.~\eqref{eq:test.3}. In the ideal case of $g\to f/I$, we expect the weight distribution to approach a delta function. In Fig.~\ref{fig:result.3}, we see that the trained results are much more like a delta function than the flat prior, showing significant improvement in the ability to draw samples from this function.

It is clear from these considerations that for the low-dimensional integrals ($D\leq 4$), all three integrators achieve reasonable results. If the target uncertainty is not very small, \veg\ or \foam\ provide the best integrator, depending on the integrand at hand. If, however, a very small target uncertainty is needed, \iflow\ is the better option as it adapts really well to the shape of the integrand. It is only the fact that \iflow\ adapts slower than \veg\ that makes \iflow\ lose in the beginning, as illustrated in Figure~\ref{fig:relerr}. For higher-dimensional integrands ($D\geq 4$) \iflow\ requires fewer function calls because it adapts better to the integrand. For example, \veg\ fails in the integration of 16-dimensional Camel function completely (missing one of the peaks) and \foam\ has a large uncertainty on the final result, even though it has much more function calls. \foam\ also performs poorly in the case of the Gaussian in 16 dimensions. In both of these cases, \foam\ approximately requires $b^D$ number of cells to map out all the features of a function, where $b$ is the average number of bins in each dimension. If $b$ is taken to be 2, for 16 dimensions, the number of cells required is at least $2^{16}$, which is far greater than our set cut-off of 10000 cells. Therefore, when dealing with high-dimensional integrals, \foam\ is the least efficient integrator. 

To quantify the computational overhead of \iflow\ in comparison to \veg, we trained both for 100 iterations with 5000 points per iteration on the polynomial function. It took \veg\ consistently 2 seconds for 2, 4, 8, 16, and 32 dimensions, and it took \iflow\ 14.7, 37.2, 80.1, 176.4, and 359.2 seconds respectively on a Laptop with Intel(R) Core(TM) i7-7700HQ CPU @ 2.80GHz. This increase is due to needing more Coupling Layers and therefore increasing the number of trainable parameters in higher dimensions. Working out the time for the 32 dimensions, we find that if the function evaluation takes much longer than 720 $\mu \text{s}$ then the overhead starts to become unimportant. Additionally, if the difference in function evaluations to reach a target precision are taken into account, the time for function evaluation is even smaller in order for the additional overhead of \iflow\ to become insignificant.

To summarize, \iflow\ provides the best integrator for integrals in 4 or more dimensions, especially if a high precision is needed and/or the integrand is numerically expensive and slow to evaluate. 

\begin{figure}
    \centering
    \begin{subfigure}[b]{0.4\textwidth}
        \includegraphics[scale=0.65]{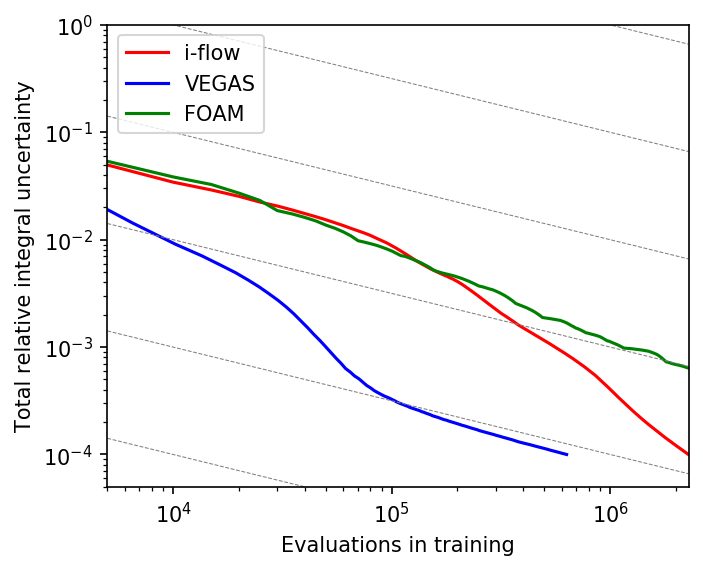}
        \caption{4-dimensional Gaussian}
        \label{fig:4dgauss}
    \end{subfigure}
    \hspace{1em}
    \vspace{1em}
    \begin{subfigure}[b]{0.4\textwidth}
        \includegraphics[scale=0.65]{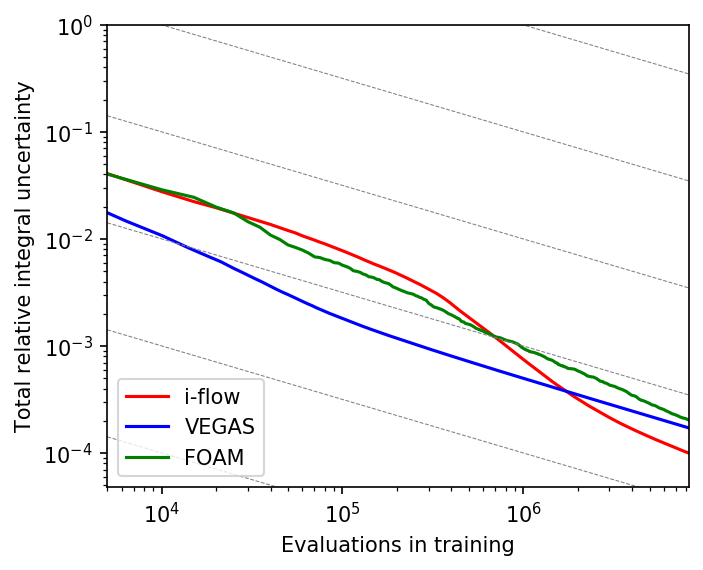}
        \caption{4-dimensional Camel}
        \label{fig:4dcamel}
    \end{subfigure}
    \hspace{1em}
    \vspace{1em}
    \begin{subfigure}[b]{0.4\textwidth}
        \includegraphics[scale=0.65]{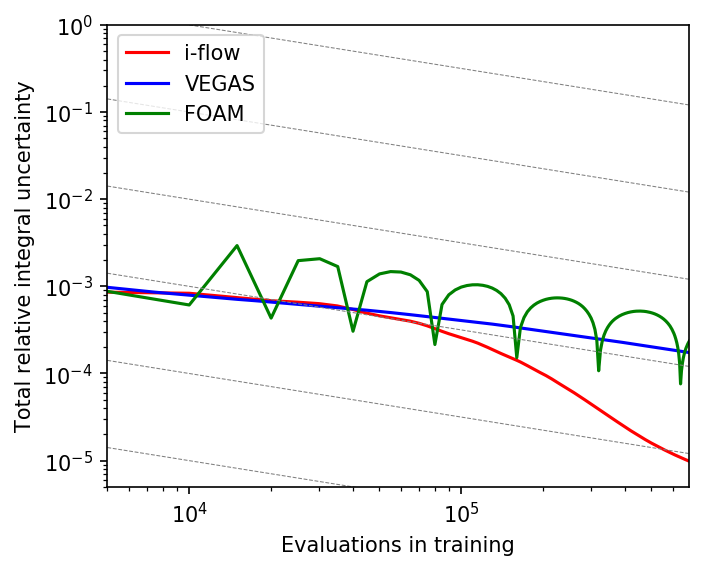}
        \caption{54-dimensional Polynomial}
        \label{fig:poly}
    \end{subfigure}
    \hspace{1em}
    \vspace{1em}
    \begin{subfigure}[b]{0.4\textwidth}
        \includegraphics[scale=0.65]{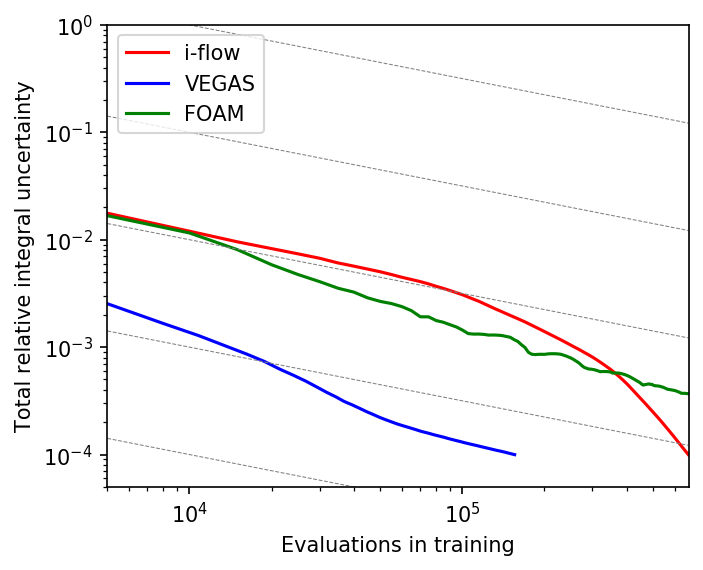}
        \caption{Scalar top loop}
        \label{fig:box}
    \end{subfigure}
    \caption{Accumulated (total) relative uncertainty of the integral, defined as the inverse-variance weighted combination of the uncertainties (normalized to the integral value) per iteration, as a function of number of points used in training for four different test functions. The dashed lines indicate a $1/\sqrt{N}$ scaling.}\label{fig:relerr}
\end{figure}

\begin{figure}[t]
  \centerline{\includegraphics[width=0.5\textwidth]{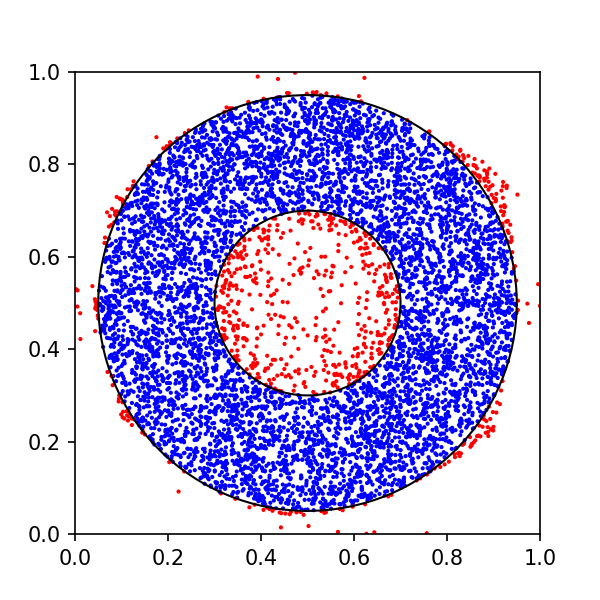}}
  \caption{A set of 7500 points sampled after training \iflow\ with 5M points on the Ring function. 6720 are inside (blue), 780 outside (red). }
  \label{fig:result.4}
\end{figure}

\begin{figure}[t]
  \centerline{\includegraphics[width=0.5\textwidth]{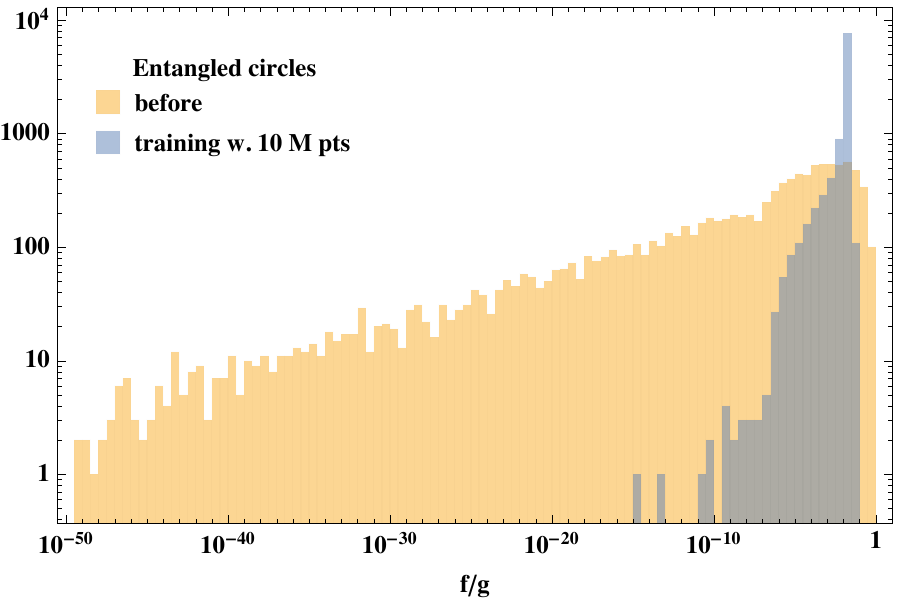}}
  \caption{Weights of 10000 points, sampled after training \iflow~on Entangled circles (\ref{eq:test.3}). $g$ is a flat distribution before training and approximately resembles the shape of $f$ after training.}
  \label{fig:result.3}
\end{figure}

\section{Conclusion and Outlook}
\label{sec:conclusions}
As shown in the previous section, \iflow\ tends to do better than both \veg\ and \foam\ for all the test cases provided. However, \iflow\ comes with a few downsides. Since \iflow\ has to learn all the correlations of the function, it takes significantly longer to achieve optimal performance compared to the other integrators. This can be seen in Fig.~\ref{fig:relerr}. This obviously translates to longer training times. Additionally, the memory footprint required for \iflow\ is much larger due to requiring storage for quicker parameter updates within the NNs. Both of these can be overcome with future improvements.

There are several directions in which we plan to improve the presented setup in the future. So far, we only used simple NN architectures in our coupling layers. Using convolutional NNs instead might improve the convergence of the normalizing flow for complicated integrands, as these networks have the ability to learn complicated shapes in images with fewer parameters than dense networks. 

The setup suggested in~\cite{2019arXiv190510347T} would allow the extension of \iflow\ to discrete distributions, which also has applications in HEP~\cite{Gao:2020zvv,Bothmann:2020ywa}. Another way to implement this type of information is utilizing Conditional Normalizing Flows~\cite{winkler2019learning}.

The implementation of transflow-learning, which was suggested in~\cite{gambardella2019transflow}, would allow the use of a trained normalizing flow on different, but similar problems without retraining the network. Such problems arise in HEP when new-physics effects from high energy scales modify scattering properties at low energies slightly and are described in an effective field theory framework. Another application for transflow-learning would be to train one network for a given dimensionality and adapt the network for another problem with the same dimensionality.

Using techniques like gradient checkpointing~\cite{2016arXiv160406174C} have the potential to reduce the memory usage substantially, therefore allowing more points to be used at each training step or larger NN architectures. 

The setup presented in~\cite{2018arXiv180703039K}, which introduces invertible $1\times 1$ convolutions, showed an improved performance over the vanilla implementation of the normalizing flows, which possibly also applies to our case. These $1\times1$ convolutions are generalizations of permutation operators acting on the inputs. Additionally, this would modify the maximum number of coupling layers required by having more expressive permutations.

%= Acknowledgements ============================================
\section*{Acknowledgements}
We thank Joao M. Goncalves Caldeira, Felix Kling, Luisa Lucie-Smith, Tilman Plehn, Holger Schulz, Nhan Tran, Paddy Fox, William Jay, David Shih, and the participants of the Aspen workshop ``The Energy Frontier Beyond the LHC Run 2” for their comments and discussions. We further thank Stefan H\"oche for helpful discussions, comments, and for his \foam\ implementation~\cite{Hoeche:foam}. \\
This manuscript has been authored by Fermi Research Alliance, LLC under Contract No. DE-AC02-07CH11359 with the U.S. Department of Energy, Office of Science, Office of High Energy Physics. This work was performed in part at Aspen Center for Physics, which is supported by National Science Foundation grant PHY-1607611. C.K. acknowledges the support of the Alexander von Humboldt Foundation.

\begin{appendices}
\section{Coupling Layer details}
\label{sec:app:CLdetails}
The implementation of the layers available in \iflow\ are detailed below. The layers are based on the work of~\cite{DBLP:journals/corr/abs-1808-03856, 2019arXiv190604032D} and are reproduced here for the convenience of the reader.

\subsection{Piecewise Linear}
For the piecewise linear coupling layer~\cite{DBLP:journals/corr/abs-1808-03856}, given $K$ bins of width $w$, the probability density function (PDF) is defined as:
  \begin{equation}
      q_i(t) = \begin{cases}
      Q_{i1}/w & t < w \\
      Q_{i2}/w & w \leq t < 2w \\
      \quad\vdots & \\
      Q_{iK}/w & 1-w \leq t < 1 \\
      \end{cases}\;.
  \end{equation}
  The cumulative distribution function (CDF) is defined by the integral giving:
  \begin{equation}
      C(x_i^B;Q) = \alpha Q_{ib} + \sum_{k=1}^{b-1} Q_{ik},
  \end{equation}
  where $b$ is the bin in which $x_i^B$ occurs ($(b-1)w\leq x_i^B < bw$), and $\alpha=\frac{x_i^B-(b-1)w}{w}$. Alternatively, we can define $b$ as the maximal $b$ for which  $\left(C_i-\sum_{k=1}^{b-1}Q_{ik}\right)>0$. The inverse CDF is given by:
  \begin{equation}
      x_i^B(C_i;Q) = \frac{w\left(C_i-\sum_{k=1}^{b-1}Q_{ik}\right)}{Q_{ib}}+(b-1)w.
  \end{equation}
  
  The Jacobian for this network is straightforward to calculate, and gives:
  \begin{equation}
      \left|\frac{\partial C}{\partial x_B}\right| = \prod_i Q_{ib}/w\;.
  \end{equation}

The piecewise linear layers require fixed bin widths in each layer. For details on why this is required, see Appendix B of~\cite{DBLP:journals/corr/abs-1808-03856}. 

\subsection{Piecewise Quadratic}

  For the piecewise quadratic coupling layer~\cite{DBLP:journals/corr/abs-1808-03856}, given $K$ bins with widths $W_{ik}$, with $K+1$ vertex heights given by $V_{ik}$, the PDF is defined as:
  \begin{equation}
      q_i(t) = \begin{cases}
      \frac{V_{i2}-V_{i1}}{W_{i1}}t+V_{i1} & t < W_{i1} \\
      \frac{V_{i3}-V_{i2}}{W_{i2}}\left(t-W_{i1}\right)+V_{i2} & W_{i1} \leq t < W_{i1}+W_{i2} \\
      \quad\vdots & \\
      \frac{V_{i(K+1)}-V_{iK}}{W_{iK}}\left(t-\sum_{k=1}^{K-1}W_{ik}\right)+V_{iK} & \sum_{k=1}^{K-1} W_{ik} \leq t < 1 \\
      \end{cases}
  \end{equation}
  Integrating the above equation leads to the CDF:
  \begin{equation}
      C(x_i^B;W,V) =\frac{\alpha^2}{2}\left(V_{ib+1}-V_{ib}\right) W_{ib} + V_{ib}W_{ib}\alpha + \sum_{k=1}^{b-1}\frac{V_{ik+1} + V_{ik}}{2} W_{ik},
  \end{equation}
  where $b$ is defined as the solution to~\footnote{Note that this definition means $b\in [1,K]$.} $\sum_{k=1}^{b-1} W_{ik} \leq x_i^B < \sum_{k=1}^{b} W_{ik}$, and $\alpha=\frac{x_i^B-\sum_{k=1}^{b-1}W_{ik}}{W_{ib}}$ is the relative position of $x_i^B$ in bin $b$. Inverting the CDF leads to:
  \begin{equation}
      x_i^B(C_i;W,V) = W_{ib}\left(\frac{-V_{ib}}{V_{ib+1}-V_{ib}}+\sqrt{\frac{V_{ib}^2}{\left(V_{ib+1}-V_{ib}\right)^2}+2\beta}\right) + \sum_{k=1}^{b-1} W_{ik},
  \end{equation}
  where $b$ is defined as the solution to 
  \begin{equation}
\sum_{k=1}^{b-1} \frac{V_{ik}+V_{ik+1}}{2} W_{ik} \leq C_i < \sum_{k=1}^{b} \frac{V_{ik}+V_{ik+1}}{2} W_{ik},
  \end{equation}
  and $\beta$ is the relative position of $C_i$ in the bin $b$, and is given by:
  \begin{equation}
    \beta = \frac{C_i-\sum_{k=1}^{b-1}\frac{V_{ik}+V_{ik+1}}{2}W_{ik}}{\left(V_{ib+1}-V_{ib}\right)W_{ib}}.
  \end{equation}

\subsection{Piecewise Rational Quadratic}\label{sec:app.PRQ}
  For the piecewise rational quadratic coupling layer~\cite{2019arXiv190604032D}, given $K+1$ knot points $\left\{\left(x^{(k)},y^{(k)}\right)\right\}_{k=0}^{K}$ that are monotonically increasing, with $(x^{(0)},y^{(0)})=(0,0)$ and $(x^{(K)},y^{(K)})=(1,1)$, and $K+1$ non-negative derivatives $\left\{d^{(k)}\right\}_{k=0}^{K}$, the CDF can be calculated using the algorithm from~\cite{10.1093/imanum/2.2.123}, which is roughly reproduced below.
  
  First, we define the bin widths ($w^{(k)} = x^{(k+1)}-x^{(k)}$) and the slopes ($s^{(k)} = \frac{y^{(k+1)}-y^{(k)}}{w^{(k)}}$). We next obtain the fractional distance ($\xi$) between the two knots that the point of interest ($x$) lies ($\xi = \frac{x - x^{(k)}}{w^{(k)}}$, where $k$ is the bin $x$ lies in). The CDF is given by:
  \begin{equation}
      g(x) = \frac{\alpha(\xi)}{\beta(\xi)},
  \end{equation}
  where the details of $\alpha(\xi)$ and $\beta(\xi)$ can be found in~\cite{10.1093/imanum/2.2.123}, but simplifies to:
  \begin{equation}
      g(x) = y^{(k)} + \frac{\left(y^{(k+1)} - y^{(k)}\right)\left[s^{(k)}\xi^2+d^{(k)}\xi\left(1-\xi\right)\right]}{s^{(k)}+\left[d^{(k+1)}+d^{(k)}-2s^{(k)}\right]\xi\left(1-\xi\right)},
  \end{equation}
  which is noted to be less prone to numerical issues~\cite{10.1093/imanum/2.2.123}. The inverse can be found by solving a quadratic equation~\cite{2019arXiv190604032D}:
  \begin{equation}
      q(x) = \alpha(\xi)-y\beta(\xi)=ax^2+bx+c=0,
  \end{equation}
  where the coefficients are given in~\cite{2019arXiv190604032D}, solving this equation for the solution that gives a monotonically increasing $x$ results in:
  \begin{equation}
      x=\frac{-b+\sqrt{b^2-4ac}}{2a} = \frac{2c}{-b-\sqrt{b^2-4ac}},
  \end{equation}
  where the second form is numerically more precise when $4ac$ is small, and is also valid for $a=0$~\cite{2019arXiv190604032D}.

\section{Loss functions}
\label{sec:app.loss}
We implemented several different divergences that can be used as loss functions. They differ in $p\leftrightarrow q$ symmetry, relative weight between small and large deviations, treatment of $p=0$ case (also in the derivative), and numerical complexity. All of them are from the class of $f$-divergences~\cite{2014ISPL...21...10N}. \\
Pearson $\chi^2$ divergence:
\begin{equation}
    D_{\chi^2} = \int \frac{(p(x)-q(x))^2}{q(x)}\; dx
\end{equation}
Kullback–Leibler divergence:
\begin{equation}
    D_{KL} = \int p(x) \log{\left(\frac{p(x)}{q(x)}\right)}\; dx
\end{equation}
squared Hellinger distance: 	
\begin{equation}
    D_{H^2} = \int 2 \left(\sqrt{p(x)}-\sqrt{q(x)}\right)^2\; dx
\end{equation}
Jeffreys divergence:
\begin{equation}
    D_{J} = \int (p(x)-q(x))(\log{p(x)}-\log{q(x)}) \; dx
\end{equation}
Chernoff's $\alpha$-divergence: 	
\begin{equation}
    D_{C\alpha} = \frac{4}{1-\alpha^2} \left(1-\int p(x)^{\frac{1-\alpha}{2}} q(x)^{\frac{1+\alpha}{2}}\; dx\right)
\end{equation}
exponential divergence: 	
\begin{equation}
\label{eq:exp.div}
    D_{e} = \int p(x) \log{\left(\frac{p(x)}{q(x)}\right)}^2\; dx
\end{equation}
$(\alpha,\beta)$-product divergence: 
\begin{equation}
    D_{\alpha\beta} = \frac{2}{(1-\alpha)(1-\beta)}\int\left(1-\left(\frac{q(x)}{p(x)}\right)^{\frac{1-\alpha}{2}}\right)\left(1-\left(\frac{q(x)}{p(x)}\right)^{\frac{1-\beta}{2}}\right) p(x)\; dx
\end{equation}
Jensen–Shannon divergence:
\begin{equation}
    D_{JS} = \frac{1}{2}\int p(x) \log{\left(\frac{2p(x)}{p(x)+q(x)}\right)}+q(x) \log{\left(\frac{2q(x)}{p(x)+q(x)}\right)}\; dx
\end{equation}

\section{Sector decomposition of scalar loop integrals}
\label{sec:app.box}
Following~\cite{Binoth:2002xh}, we give the integral representations of triangle and box functions in 4 dimensions using the Feynman parametrisation. To begin with, the triangle integral with external particles of energy $\sqrt{s_1},\sqrt{s_2},\sqrt{s_3}$ and internal propagators of masses $m_1,m_2,m_3$ is given by
\begin{equation}
\begin{split}
    I_3(s_1,s_2,s_3,m_1^2,m_2^2,m_3^2)=&\int \frac{d^4k}{i\pi ^2}\frac1{[(k-r_1)^2-m_1^2][(k-r_2)^2-m_2^2][k^2-m_3^2]}\\
    =&-\int_{0}^{\infty}d^3x\,\delta(1-x_{123})\frac{x_{123}^{-1}}{\mathcal{F}_{Tri}}\\
    \mathcal{F}_{Tri}=(-s_1)x_1x_3+(-s_2)&x_1x_2+(-s_3)x_2x_3+x_{123}(x_1m_1^2+x_2m_2^2+x_3m_3^2)-i\epsilon \\
    x_{123}=x_1+x_2+x_3\quad\quad\quad&
\end{split}
\end{equation}
The 3-dimensional integral is further split into 3 sectors by the decomposition:
\begin{equation}
    1=\Theta(x_1>x_2,x_3)+\Theta(x_2>x_1,x_3)+\Theta(x_3>x_1,x_2)
\end{equation}
For example, when $x_3>x_1,x_2$, after the variable transformation $t_i=x_i/x_3(i=1,2)$, the integral simplifies to
\begin{equation}
\begin{split}
S_{Tri}&(s_1,s_2,s_3,m_1^2,m_2^2,m_3^2)=\int_0^1 dt_1dt_2 \frac{(1+t_1+t_2)^{-1}}{\tilde{\mathcal{F}}_{Tri}(s_1,s_2,s_3,m_1^2,m_2^2,m_3^2,t_1,t_2)}\\
&\tilde{\mathcal{F}}_{Tri}(s_1,s_2,s_3,m_1^2,m_2^2,m_3^2,t_1,t_2)=(-s_1)t_1+(-s_2)t_1t_2+(-s_3)t_2\\
&\qquad\qquad\qquad\qquad\qquad\qquad\qquad+(1+t_1+t_2)(t_1m_1^2+t_2m_2^2+m_3^2)
\end{split}
\end{equation}
Therefore,
\begin{equation}
    \begin{split}
        I_3=S_{Tri}(s_1,s_2,s_3,m_1^2,m_2^2,m_3^2)+ S_{Tri}(1\, 2\, 3\,\to 2\, 3\, 1\,)+S_{Tri}(1\, 2\, 3\,\to 3\, 1\, 2\,)
    \end{split}
\end{equation}

One can perform the same trick to treat the box integral with 4 external fields and 4 propagators. After sector decomposition, one gets
\begin{equation}
\begin{split}
    I_4=&S_{Box}(s_{12},s_{23},s_1,s_2,s_3,s_4,m_1^2,m_2^2,m_3^2,m_4^2)+ \textbf{permutations}\\
    S_{Box}&(s_{12},s_{23},s_1,s_2,s_3,s_4,m_1^2,m_2^2,m_3^2,m_4^2)=\int_0^1 dt_1dt_2dt_3 \frac1{\tilde{\mathcal{F}}_{Box}^2}\\
    \tilde{\mathcal{F}}&_{Box}=(-s_{12})t_2+(-s_{23})t_1t_3+(-s_1)t_1+(-s_{2})t_1t_2+(-s_{3})t_2t_3\\
    &\quad\quad+(-s_{4})t_3+(1+t_1+t_2+t_3)(t_1m_1^2+t_2m_2^2+t_3m_3^2+m_4^2)
\end{split}
\end{equation}

\end{appendices}
\clearpage

\bibliography{ML_PT}

\begin{thebibliography}{10}
\bibitem{fn:camel16}
{The estimate of \iflow\ only deviates from the true value because the
  estimates from all iterations are combined and the first $\sim 200$ epochs
  only ``see" one of the two peaks. Combining 15 of the last epochs yields
  $0.86377(136)$, which is closer to the true value.}
\end{thebibliography}
\bibliographystyle{unsrturl}

\end{document}